\documentclass{sigplanconf}

\usepackage{amsmath,amssymb,color,infer,mathtools,mathwidth,stmaryrd,url}
\usepackage[shortlabels]{enumitem}

\usepackage{amsthm}
\newtheorem{theorem}{Theorem}
\newtheorem{lemma}[theorem]{Lemma}

\theoremstyle{definition}
\newtheorem{defn}[theorem]{Definition}

\newcommand{\seq}[1]{\overline{#1}}

\usepackage[numbers,sort&compress,square]{natbib}

\usepackage{float}
\floatstyle{boxed}
\restylefloat{figure}

\makeatletter
\def\operator@font{\sf}
\makeatother

\usepackage{zi4}
\DeclareMathAlphabet{\mathtt}{OT1}{zi4}{m}{n}

\usepackage{listings}
\lstnewenvironment{code}
    {\lstset{}%
      \csname lst@SetFirstLabel\endcsname}
    {\csname lst@SaveFirstLabel\endcsname}
\lstMakeShortInline{!}

\newcommand{\T}{\mathcal{T}}
\newcommand{\sem}[1]{\ensuremath{\llbracket #1 \rrbracket}}
\newcommand{\interp}[1]{\ensuremath{\T\sem{#1}}}
\newcommand{\intype}[1]{\ensuremath{\T^{\sf type}\sem{#1}}}
\newcommand{\interm}[1]{\ensuremath{\T^{\sf term}\sem{#1}}}
\newcommand{\intsch}[1]{\ensuremath{\T^{\sf scheme}\sem{#1}}}
\def\tr{\leadsto}
\def\pto{\rightharpoonup}
\newcommand{\secref}[1]{(\S\ref{sec:#1})}
\newcommand{\entails}{\ensuremath{\Vdash}}
\newcommand{\predh}[2]{\ensuremath{#1\ #2}}

\newcommand{\clause}[3]{\ensuremath{#1 : #2 \then #3}}
\newcommand{\qclause}[4]{\ensuremath{#1 : \forall #2. \: #3 \then #4}}
\newcommand{\then}{\ensuremath{\Rightarrow}}
\newcommand{\tuple}[1]{\langle #1 \rangle}
\newcommand{\fd}[3]{\ensuremath{#1 : #2 \tr #3}}
\newcommand{\nono}[1]{\ensuremath{\vdash non\text{-}overlapping(#1)}}
\newcommand{\cov}[1]{\ensuremath{\vdash covering(#1)}}
\newcommand{\restrict}[2]{\ensuremath{#1|_{#2}}}
\newcommand{\mkwd}[1]{\ensuremath{\text{\texttt{\underbar{#1}}}}}
\newcommand{\gr}[1]{\lfloor #1 \rfloor}
\newcommand{\oml}[0]{$\mathrm{H}^-$}
\newcommand{\omlhead}[0]{$\mathbf{H}^-$}

\newcommand{\De}{\ensuremath{\Delta}}

\newcommand{\trule}[1]{\textsc{(#1)}}
\newcommand{\erule}[1]{\textsc{\{#1\}}}
\newcommand{\I}[1]{\ensuremath{#1\!\!}~I}
\newcommand{\E}[1]{\ensuremath{#1\!\!}~E}
\newcommand{\Iforall}{\ensuremath{\forall\!}~I}
\newcommand{\Eforall}{\ensuremath{\forall\!}~E}
\newcommand{\isp}{\hspace{\infskip}}

\newcommand{\eql}[6]{\ensuremath{#1 \mid #2 \vdash_{#3} #4 \equiv #5: #6}}
\newcommand{\eqs}[3]{\eql{P}{\Gamma}{\Psi}{#1}{#2}{#3}}
\newcommand{\eqx}[3]{\ensuremath{\vdash_\Psi #1 \equiv #2: #3}}
\newcommand{\typ}[5]{\ensuremath{#1 \mid #2 \vdash_{#3} #4 : #5}}
\newcommand{\tys}[2]{\typ{P}{\Gamma}{A}{#1}{#2}}

\DeclareMathOperator{\dom}{dom}
\DeclareMathOperator{\fix}{fix}

\exclusivelicense
\conferenceinfo{Haskell~'14}{September 6, 2014,
Gothenburg, Sweden}
\CopyrightYear{2014}
\copyrightdata{978-1-4503-3041-1/14/09}
\doi{2633357.2633364}

\title{A Simple Semantics for Haskell Overloading}
\authorinfo{J. Garrett Morris}
           {University of Edinburgh}
           {Garrett.Morris@ed.ac.uk}

\begin{document}

\maketitle

\begin{abstract}
  As originally proposed, type classes provide overloading and ad-hoc definition, but can still be
  understood (and implemented) in terms of strictly parametric calculi.  This is not true of
  subsequent extensions of type classes.  Functional dependencies and equality constraints allow the
  satisfiability of predicates to refine typing; this means that the interpretations of equivalent
  qualified types may not be interconvertible.  Overlapping instances and instance chains allow
  predicates to be satisfied without determining the implementations of their associated class
  methods, introducing truly non-parametric behavior.  We propose a new approach to the semantics of
  type classes, interpreting polymorphic expressions by the behavior of each of their ground
  instances, but without requiring that those behaviors be parametrically determined.  We argue that
  this approach both matches the intuitive meanings of qualified types and accurately models the
  behavior of programs.
\end{abstract}



\section{Introduction}\label{sec:introduction}

Implicit polymorphism (as provided by the Hindley-Milner type systems in ML and Haskell) provides a
balance between the safety guarantees provided by strong typing, and the convenience of generic
programming.  The Hindley-Milner type system is strong enough to guarantee that the evaluation of
well-typed terms will not get stuck, while polymorphism and principal types allow programmers to
reuse code and omit excessive type annotation.  Type classes~\cite{WadlerBlott89} play a similar
role for overloading: they preserve strong typing (ruling out run-time failures from the use of
overloaded symbols in undefined ways) without requiring that programmers explicitly disambiguate
overloaded expressions.  Since their introduction, type classes have seen numerous extensions, such
as multi-parameter type classes, functional dependencies~\cite{Jones00}, and overlapping
instances~\cite{PeytonJones97}; a variety of practical uses, from simple overloading to capturing
complex invariants and type-directed behavior; and, the adoption of similar approaches in other
strongly-typed programming languages, including Isabelle and Coq.

\subsection{Dictionary-Passing and its Disadvantages}

The semantics of type classes has primarily been given by translations from instance declarations
and (implicit) overloading to dictionaries and (explicit) dictionary arguments.  This parallels the
treatment of implicit polymorphism by translation to systems with explicit polymorphism (such as
System~F), and shares similar challenges.  For a simple example, in Haskell, the !map! function has
the polymorphic type scheme $(t \to u) \to [t] \to [u].$ In translating to System~F, this could be
interpreted as either
\[
  \forall t. \forall u. (t \to u) \to [t] \to [u] \quad \text{or} \quad
  \forall u. \forall t. (t \to u) \to [t] \to [u].
\]
But these types are not equivalent: they express different orders of passing type arguments.  There
are various ways of addressing this discrepancy: for example, Mitchell~\cite{Mitchell88} shows that,
for any two translations of an implicitly typed scheme, there is a term (which he calls a retyping
function) which transforms terms of one translation to terms of the other, while only manipulating
type abstractions and applications.  Similar issues arise in the semantics of type classes.  For
example, a function to compare pairs $(t, u)$ for equality could be given either the type scheme
\[
  (\predh{\mathtt{Eq}}{t}, \predh{\mathtt{Eq}}{u}) \then (t, u) \to (t, u) \to \mathtt{Bool}
\]
or the type scheme
\[
  (\predh{\mathtt{Eq}}{u}, \predh{\mathtt{Eq}}{t}) \then (t, u) \to (t, u) \to \mathtt{Bool}.
\]
In a dictionary-passing translation, type classes are interpreted by tuples, called dictionaries,
containing the type-specific implementations of each of the class methods.  Class instances
correspond to dictionary definitions, while predicates in types correspond to dictionary arguments.
In the case of the !Eq! class, which has equality and inequality methods, we could define !Eq!
dictionaries by
\[
  \mathtt{EqDict}\, t = (t \to t \to \mathtt{Bool}, t \to t \to \mathtt{Bool}).
\]
Even though the two types for pair equality above are equivalent in the implicitly overloaded
setting, their dictionary-passing translations are not: the first corresponds to a function of type
\[
  \mathtt{EqDict}\,t \to \mathtt{EqDict}\,u \to (t, u) \to (t, u) \to \mathtt{Bool},
\]
while the second corresponds to
\[
  \mathtt{EqDict}\,u \to \mathtt{EqDict}\,t \to (t, u) \to (t, u) \to \mathtt{Bool},
\]
Again, approaches exist to address this discrepancy: for example, Jones
shows~\cite{Jones93Coherence} that there are conversion functions, similar to Mitchell's retyping
functions, to convert between different translations of the same overloaded term.

Our own work began by exploring instance chains~\cite{Morris10}, a proposed extension to
Haskell-like type class systems.  In the course of this exploration, we discovered several
difficulties with existing approaches to the semantics of overloading.

\paragraph{Mismatch in expressivity.}

System~F typing is significantly more expressive than the Hindley-Milner type systems it is used to
model.  In particular, even within the translation of an ML or Haskell type scheme, there are
arbitrarily many expressions that do not correspond to any expressions of the source language.  The
problem is compounded when considering dictionary-passing translations of type classes.  For
example, there is no notion in Haskell of class instances depending on terms; on the other hand,
there is no difficulty in defining a term of type $\mathtt{Int} \to \mathtt{EqDict\,Int}$.  Uses of
such a term cannot be equivalent to any use of the methods of !Eq!.  As a consequence, there are
properties of source programs (for example, that any two instances of $\mathtt{=\!=}$ at the same
type are equal) that may not be provable of their dictionary-passing translation without reference
to the specific mechanisms of translation.

\paragraph{Predicates refine typing.}

\newcommand{\Elems}[2]{\predh{\mathtt{Elems}}{#1\,#2}}

Second, the notions of equivalence of System~F and Haskell types diverge once the satisfiability of
predicates can refine typing.  For example, functional dependencies allow programmers to declare
that some parameters of a class depend upon others; in the declaration
\begin{code}
class Elems c e | c -> e where
  empty :: c
  insert :: e -> c -> c
\end{code}
the dependency !c -> e! captures the intuition that the type of a container's elements are
determined by the type of the container. Concretely, given two predicates $\Elems{\tau}{\upsilon}$
and $\Elems{\tau'}{\upsilon'}$, if we know that $\tau = \tau'$, then we can conclude $\upsilon =
\upsilon'$.  This property is lost in the dictionary-passing translation.  Dictionaries for !Elems!
contain just their methods:
\[
  \mathtt{ElemsDict} \, c \, e = (c, e \to c \to c)
\]
As a consequence, there are types that are equivalent in Haskell, but are not interconvertible in
the dictionary-passing interpretation.  For example, the type
\(
  (\Elems{c}{e}, \Elems{c}{e'}) \then e \to e' \to c
\)
is equivalent to the (simpler) type
\(
  (\Elems{c}{e}) \then e \to e \to c
\)
as we must have that $e = e'$ for the qualifiers in the first type to be satisfiable.  However,
there is no corresponding bijection between terms of type
\(
  \mathtt{ElemsDict}\,c\,e \to \mathtt{ElemsDict}\,c\,e' \to e \to e' \to c
\)
and terms of type
\(
  \mathtt{ElemsDict}\,c\,e \to e \to e \to c.
\)
While we can construct a term of the second type given a term of the first, there is no parametric
construction of a term of the first type from a term of the second.

\paragraph{Non-parametric behavior.}

Finally, other extensions to class systems make it possible to define terms which have no
translation to parametric calculi.  For example, we could define a function !invBool! that negated
booleans and was the identity on all other types.  We begin by introducing a suitable class:
\begin{code}
class Univ t where
  invBool :: t -> t
\end{code}
There are several approaches to populating the class, using different extensions of the Haskell
class system.  Using overlapping instances~\cite{PeytonJones97}, we could simply provide the two
desired instances of the class, relying on the type checker to disambiguate them based on their
specificity:
\begin{code}
instance Univ Bool where
  invBool = not
instance Univ t where
  invBool = id
\end{code}
Using instance chains, we would specify the ordering directly:
\begin{code}
instance Univ Bool where
  invBool = not
else Univ t where
  invBool = id
\end{code}
With either of these approaches, we might expect that the type of the class method !invBool! is
$(\predh{\mathtt{Univ}}{t}) \then t \to t.$ However, the predicate $\predh{\mathtt{Univ}}{\tau}$ is
provable for arbitrary types $\tau$.  Thus, the above type is intuitively equivalent to the
unqualified type $t \to t$; however, there is no term of that type in a purely parametric calculus
that has the behavior of method !invBool!.  (In practice, this is avoided by requiring that
!invBool!'s type still include the !Univ! predicate, even though it is satisfied in all possible
instantiations; while this avoids the difficulties in representing !invBool! in a parametric
calculus, it disconnects the meaning of qualified types from the satisfiability of their
predicates.)

\subsection{Specialization-based Semantics}

We propose an alternative approach to the semantics of type-class based implicit overloading.
Rather than interpret polymorphic expressions by terms in a calculus with higher-order polymorphism,
we will interpret them as type-indexed collections of (the interpretations of) monomorphic terms,
one for each possible ground instantiation of their type.  We call this a specialization-based
approach, as it relates polymorphic terms to each of their (ground-typed) specializations.  We
believe this approach has a number of advantages.
\begin{itemize}
\item First, our approach interprets predicates directly as restrictions of the instantiation of
  type variables, rather than through an intermediate translation.  Consequently, properties of the
  source language type system---such as the type refinement induced by the !Elems! predicates---are
  immediately reflected in the semantics, without requiring the introduction of coercions.
\item Second, our approach naturally supports non-parametric examples, such as class !Univ!, and
  avoids introducing artificial distinction between the semantics of expressions using parametric
  and ad-hoc polymorphism.
\item Third, because our approach does not need to encode overloading via dictionaries, it becomes
  possible to reason about class methods directly, rather than through reasoning about the
  collection of dictionaries defined in a program.
\end{itemize}
Our approach builds on Ohori's simple semantics for ML polymorphism~\cite{Ohori89}, extended by
Harrison to support polymorphic recursion~\cite{Harrison05}.

In this paper, we introduce a simple overloaded language called \oml{}~\secref{oml-eqn}, and give
typing and equality judgments in the presence of classes and class methods.  We apply our
specialization-based approach to give a denotational semantics of \oml{}~\secref{sem-over}, and show
the soundness of typing and equality with respect to the denotational semantics~\secref{formal}.  We
also develop two examples, to demonstrate the advantages of our approach.  First, we consider a pair
of definitions, one parametric and the other ad-hoc, defining operational equivalent terms.  We show
that the defined terms are related by our equality judgment~\secref{poly-id-eqn} and have the same
denotations~\secref{poly-id-den}.  This demonstrates the flexibility of our approach, and the
ability to reason about class methods directly (the second and third advantages listed above).
Second, we extend \oml{} with functional dependencies~\secref{fundeps}, and establish the soundness
of the (extended) typing and equality judgments, all without having to augment the models of terms.
This demonstrates the extensibility of our approach, and the close connection between properties of
source terms and properties of their denotations (the first advantage listed above).

\section{The \omlhead{} Language}\label{sec:oml-eqn}

\begin{figure}
\[\begin{array}{ll@{\hspace{7mm}}ll}
  \text{Term variable} & x \in Var & \text{Term constants} & k \\
  \text{Type variables} & t \in TVar & \text{Type constants} & K \\
  \text{Class names} & C & \text{Instance names} & d \in InstName \\
\end{array}\]
\[\begin{array}{lrr@{\hspace{5px}}c@{\hspace{5px}}l}
  \multicolumn{2}{l}{\text{Types}} & \tau,\upsilon & ::= & t \mid K \mid \tau \to \tau \\
  \text{Predicates} & \multicolumn{2}{r}{Pred \ni \pi} & ::= & \predh{C}{\seq\tau} \\
  \multicolumn{2}{l}{\text{Contexts}} & P,Q & ::= & \seq\pi \\
  \multicolumn{2}{l}{\text{Qualified types}} & \rho & ::= & \tau \mid \pi \then \rho \\
  \text{Type schemes} & \multicolumn{2}{r}{Scheme \ni \sigma} & ::= & \rho \mid \forall t. \sigma \\
  \text{Expressions} & \multicolumn{2}{r}{Expr \ni M,N} & ::= & x \mid k \mid \lambda x. M \mid M\,N \\
   & & &  \mid & \mu x.M \mid \mkwd{let} \; x = M \; \mkwd{in} \; N \\
  \text{Class axioms} & \multicolumn{2}{r}{Axiom \ni \alpha} & ::= & \qclause{d}{\seq{t}}{P}{\pi} \\
  \multicolumn{2}{l}{\text{Axiom sets}} & A & \subset & Axiom \\
  \text{Methods:} \\
  \multicolumn{2}{l}{\text{\ \ Signatures}} & Si & \in & Var \pto Pred \times Scheme \\
  \multicolumn{2}{l}{\text{\ \ Implementations}} & Im & \in & InstName \times Var \pto Expr \\
  \multicolumn{2}{l}{\text{Class contexts}} & \Psi & ::= & \tuple{A, Si, Im}
\end{array}\]
\caption{Types and terms of \oml{}.}\label{fig:types-terms-over}
\end{figure}

Figure~\ref{fig:types-terms-over} gives the types and terms of \oml{}; we write $\seq x$ to denote a
(possibly empty) sequence of $x$'s, and if $\pi$ is a predicate $\predh{C}{\seq\tau}$, we will
sometimes write $class(\pi)$ for $C$.  As in Jones's theory of qualified types~\cite{Jones92}, the
typical Hindley-Milner types are extended with qualified types $\rho$, capturing the use of
predicates.  We must also account for the definition of classes and their methods.  One approach
would be to expand the grammar of expressions to include class and instance declarations; such an
approach is taken in Wadler and Blott's original presentation~\cite{WadlerBlott89}.  However, this
approach makes such definitions local, in contrast to the global nature of subsequent type class
systems (such as that of Haskell), and introduces problems with principal typing (as Wadler and
Blott indicate in their discussion).  We take an alternative approach, introducing new top level
constructs (axioms~$A$, method signatures~$Si$, and method implementations~$Im$) to model class and
instance declarations.  We refer to tuples of top level information as class contexts $\Psi$, and
will give versions of both our typing and semantic judgments parameterized by such class contexts.
Note that this leaves implicit many syntactic restrictions that would be present in a full language,
such as the requirement that each instance declaration provide a complete set of method
implementations.

\subsection{\omlhead{} Typing}

\begin{figure}
\begin{gather*}
\infbox{\irule[\trule{Var}]
              {(x : \sigma) \in \Gamma};
              {\tys{x}{\sigma}}}
\isp
\infbox{\irule[\trule{$\I\to$}]
              {\typ{P}{\Gamma,x:\tau}{A}{M}{\tau'}};
              {\tys{(\lambda x.M)}{\tau \to \tau'}}}
\\
\infbox{\irule[\trule{$\E\to$}]
              {\tys{M}{\tau \to \tau'}}
              {\tys{N}{\tau}};
              {\tys{(M\,N)}{\tau'}}}
\\
\infbox{\irule[\trule{$\mu$}]
              {\typ{P}{\Gamma,x:\tau}{A}{M}{\tau}};
              {\tys{\mu x.M}{\tau}}}
\isp
\infbox{\irule[\trule{$\I\then$}]
              {\typ{P,\pi}{\Gamma}{A}{M}{\rho}};
              {\tys{M}{\pi \then \rho}}}
\\
\infbox{\irule[\trule{$\E\then$}]
              {\tys{M}{\pi \then \rho}}
              {P \entails_A \pi};
              {\tys{M}{\rho}}}
\\
\infbox{\irule[\trule{\Iforall}]
              {\tys{M}{\sigma}}
              {t \not\in ftv(\Gamma,P)};
              {\tys{M}{\forall t. \sigma}}}
\isp
\infbox{\irule[\trule{\Eforall}]
              {\tys{M}{\forall t. \sigma}};
              {\tys{M}{[\tau/t]\sigma}}}
\\
\infbox{\irule[\trule{Let}]
              {\tys{M}{\sigma}}
              {\typ{P}{\Gamma,x:\sigma}{A}{N}{\tau}};
              {\tys{(\mkwd{let}\;x = M\;\mkwd{in}\;N)}{\tau}}}
\end{gather*}
\caption{Expression typing rules of \oml{}.}\label{fig:expr-typing-oml}
\end{figure}

We begin with the typing of \oml{} expressions; our expression language differs from Jones's only in
the introduction of $\mu$ (providing recursion) .  Typing judgments take the form
\[
  \tys{M}{\sigma},
\]
where $P$ is a set of predicates restricting the type variables in $\Gamma$ and $\sigma$, and $A$ is
the set of class axioms (the latter is the only significant difference between our type system and
Jones's).  The typing rules for \oml{} expressions are given in Figure~\ref{fig:expr-typing-oml}.
We write $ftv(\tau)$ for the free type variables in $\tau$, and extend $ftv$ to predicates $\pi$,
contexts~$P$, and environments~$\Gamma$ in the expected fashion.  Rules \trule{$\then$~I} and
\trule{$\then$~E} describe the interaction between the predicate context~$P$ and qualified
types~$\rho$.  Otherwise, the rules are minimally changed from the corresponding typing rules of
most Hindley-Milner systems.

\begin{figure}
\begin{gather*}
\infbox{\irule[\trule{Assume}]
              {\pi \in P};
              {\text{-}: P \entails_A \pi}}
\\
\infbox{\irule[\trule{Axiom}]
              {(d: \forall \overline t. Q' \then \pi') \in A}
              {S \, \pi' = \pi}
              {P \entails_A S\,Q'};
              {d: P \entails_A \pi}}
\end{gather*}
\caption{Predicate entailment rules of \oml{}.}\label{fig:pred-entail}
\end{figure}

We continue with the rules for predicate entailment in \oml{}, given in
Figure~\ref{fig:pred-entail}.  The judgment $d: P \entails_A \pi$ denotes that the axiom named $d$
proves predicate $\pi$, given assumptions~$P$ and class axioms~$A$.  We use a dummy instance name,
written $\text{-}$, in the case that the goal is one of the assumptions.  We will omit the instance
name if (as in the typing rules) the particular instance used is irrelevant.  We write $P \entails_A
Q$ if there are $d_1 \dots d_n$ such that $d_i : P \entails Q_i$, and $\entails_A P$ to abbreviate
$\emptyset \entails_A P$.  Our entailment relation differs from Jones's entailment relation for type
classes and from our prior systems~\cite{Morris10} in two respects.  First, our system is
intentionally simplified (for example, we omit superclasses and instance chains).  Second, we do not
attempt to capture all the information that would be necessary for an dictionary-passing
translation; we will show that having just the first instance name is sufficient to determine the
meanings of overloaded expressions.

In the source code of a Haskell program, type class methods are specified in class and instance
declarations, such as the following:
\begin{code}
class Eq t where (==) :: t -> t -> Bool
instance Eq t => Eq [t] where xs == ys = ...
\end{code}
We partition the information in the class and instance declarations into class context tuples
$\tuple{A,Si,Im}$.  The logical content is captured by the axioms~$A$; in this example, we would
expect that there would be some instance name $d$ such that
\[
  (\qclause{d}{t}{\predh{\mathtt{Eq}}{t}}{\predh{\mathtt{Eq}}{[t]}}) \in A.
\]
Haskell's concrete syntax does not name instances; for our purposes, we assume that suitable
identifiers are generated automatically.  The method signatures are captured in the mapping $Si$; we
distinguish the class in which the method is defined (along with the corresponding type variables)
from the remainder of the method's type scheme.  For this example, we would have
\[
   Si(=\!=) = \tuple{\predh{\mathtt{Eq}}{t}, t \to t \to \mathtt{Bool}}.
\]
Note that we have not quantified over the variables appearing in the class predicate, nor included
the class predicate in the type scheme $t \to t \to \mathtt{Bool}$.  Each predicate in the range of
$Si$ will be of the form $\predh{C}{\seq{t}}$ for some class $C$ and type variables $\seq{t}$, as
they arise from class definitions.  The type scheme of a class member may quantify over variables or
include predicates beyond those used in the class itself.  For example, the !Monad! class has the
following definition:
\begin{code}
class Monad m where
  return :: a -> m a
  (>>=)  :: m a -> (a -> m b) -> m b
\end{code}
Note that the variable !a! in the type of !return! is not part of the !Monad! constraint.  Thus, we
would have that
\[
  Si(\mathtt{return}) = \tuple{\predh{\mathtt{Monad}}{m}, \forall a. a \to m\,a}.
\]
The method implementations themselves are recorded in component $Im$, which maps pairs of method and
instance names to implementing expressions.

To describe the typing of methods and method implementations, we begin by describing the type of
each method implementation.  This is a combination of the defining instance, including its context,
and the definition of the method itself.  For example, in the instance above, the body of the $=\!=$
method should compare lists of arbitrary type !t! for equality (this arises from the instance
predicate !Eq [t]! and the signature of $=\!=$), given the assumption !Eq t!  (arising from the
defining instance).  That is, we would expect it to have the type
\[
  \forall t. \predh{\mathtt{Eq}}{t} \then [t] \to [t] \to \mathtt{Bool}.
\]

We introduce abbreviations for the type scheme of each method, in general and at each instance,
assuming some class context $\tuple{A, Si, Im}$.  For each method name $x$ such that $Si(x) =
\tuple{\pi, \forall \seq{u}. \rho}$, we define the type scheme for $x$ by:
\[
  \sigma_x = \forall \seq{t}. \forall \seq{u}. \: \pi \then \rho,
\]
or, equivalently, writing $\rho$ as $Q \then \tau$:
\[
  \sigma_x = \forall \seq{t},\seq{u}. \: (\pi,Q) \then \tau
\]
where, in each case, $\seq{t} = ftv(\pi)$.  Similarly, for each method $x$ as above, and each
instance $d$ such that
\begin{itemize}
\item  $\tuple{x,d} \in \dom(Im)$;
\item  $(\qclause{d}{\seq{t}}{P}{\pi'}) \in A$; and,
\item there is some substitution $S$ such that $S\,\pi = \pi'$
\end{itemize}
we define the type scheme for $x$ in $d$ by:
\[
  \sigma_{x,d} = \forall \seq{t}, \seq{u}. \: (P, S\,Q) \then S\,\tau.
\]

\begin{figure}
\begin{gather*}
\infbox{\irule[\trule{Ctxt}]
              {\!\!\!\!
               \begin{array}{c}
                 \{ \pi \nsim \pi' \mid (d: P \then \pi), (d': P' \then \pi') \in A \}
                 \\[1px]
                 \{ (\typ{P}{\Gamma, \overline{x_i : \sigma_{x_i}}}{A}{Im(y,d)}{\sigma_{y,d}}) \mid \tuple{y,d} \in \dom(Im) \}
                 \\[1px]
                 \typ{P}{\Gamma, \overline{x_i : \sigma_{x_i}}}{A}{M}{\sigma}
               \end{array}};
              {\typ{P}{\Gamma}{\tuple{A,Si,Im}}{M}{\sigma}}}
\end{gather*}
\caption{\oml{} typing with class contexts.}\label{fig:prog-typing-oml}
\end{figure}

Finally, we give a typing rule parameterized by class contexts in Figure~\ref{fig:prog-typing-oml};
in $\overline{x_i:\sigma_{x_i}}$, the $x_i$ range over all methods defined in the program (i.e., over
  the domain of $Si$).  Intuitively, an expression $M$ has type $\tau$ under $\tuple{A,Si,Im}$ if:
\begin{itemize}
\item None of the class instances overlap.  More expressive class systems will require more
  elaborate restrictions; we give an example when extending \oml{} to support functional
  dependencies~\secref{fundeps}.
\item Each method implementation $Im(x,d)$ has the type $\sigma_{x,d}$ (methods are allowed to be
  mutually recursive).
\item The main expression has the declared type $\sigma$, given that each class method $x_i$ has
  type $\sigma_{x_i}$.
\end{itemize}

\subsection{Equality of \omlhead{} Terms}

In this section, we give an axiomatic presentation of equality for \oml{} terms.  Our primary
concerns are the treatment of polymorphism and class methods; otherwise, \oml{} differs little from
standard functional calculi.  As described in the introduction, our intention is to permit reasoning
about class methods directly, without relying on either a dictionary-passing translation or a
preliminary inlining step that resolves all method overloading.  This results in two unusual aspects
of our rules:
\begin{itemize}
\item While our presentation gives equality for expressions, it relies critically on components of
  the class context $\tuple{A,Si,Im}$---the axioms~$A$ to determine which instance solves given
  constraints, and the method implementations $Im$ to determine the behavior of methods.
\item The treatment of polymorphism cannot be completely parametric, and different equalities may be
  provable for the same term at different types; for example, we cannot hope to have uniform proofs
  of properties of the $=\!=$ method when it is defined differently at different types.
\end{itemize}

\begin{figure}
\begin{gather*}
\infbox{\irule[\erule{$\beta$}]
              {\typ{P}{\Gamma,x:\tau}{\Psi}{M}{\tau'}}
              {\typ{P}{\Gamma}{\Psi}{N}{\tau}};
              {\eqs{(\lambda x.M) N}{[N/x]M}{\tau'}}}
\\
\infbox{\irule[\erule{$\eta$}]
              {\typ{P}{\Gamma}{\Psi}{M}{\tau \to \tau'}}
              {x \not\in fv(M)};
              {\eqs{\lambda x.(M x)}{M}{\tau \to \tau'}}}
\\
\infbox{\irule[\erule{$\mu$}]
              {\typ{P}{\Gamma,x:\tau}{\Psi}{M}{\tau}};
              {\eqs{\mu x.M}{[\mu x.M/x]M}{\tau}}}
\\
\infbox{\irule[\erule{Let}]
              {\typ{P}{\Gamma}{\Psi}{M}{\sigma}}
              {\typ{P}{\Gamma,x:\sigma}{\Psi}{N}{\tau}};
              {\eqs{(\mkwd{let}\;x = M\;\mkwd{in}\;N)}{[M/x]N}{\tau}}}
\\
\infbox{\irule[\erule{Method}]
              {Si(x) = \tuple{\pi,\sigma}}
              {d: P \entails S\,\pi};
              {\eql{P}{\Gamma}{\tuple{A,Si,Im}}{x}{Im(x,d)}{S\,\sigma}}}
\\
\infbox{\irule[\erule{\Iforall}]
              {t \not\in ftv(P,\Gamma)}
              {\{ (\eqs{M}{N}{[\tau/t]\sigma}) \mid \tau \in GType \}};
              {\eqs{M}{N}{\forall t. \sigma}}}
\\
\infbox{\irule[\erule{\Eforall}]
              {\eqs{M}{N}{\forall t. \sigma}};
              {\eqs{M}{N}{[\tau/t] \sigma}}}
\\
\infbox{\irule[\erule{$\I\then$}]
              {\eql{P,\pi}{\Gamma}{\Psi}{M}{N}{\rho}};
              {\eqs{M}{N}{\pi \then \rho}}}
\\
\infbox{\irule[\erule{$\E\then$}]
              {\eqs{M}{N}{\pi \then \rho}}
              {P \entails \pi};
              {\eqs{M}{N}{\rho}}}
\end{gather*}
\caption{Equality for \oml{} terms.}\label{fig:oml-eqn}
\end{figure}

Equality judgments take the form $\eqs{M}{N}{\sigma}$, denoting that, assuming predicates~$P$,
variables typed as in $\Gamma$, and class context $\Psi$, expressions $M$ and $N$ are equal at type
$\sigma$.  To simplify the presentation, we have omitted equational assumptions; however, extending
our system with assumptions and a corresponding axiom rule would be trivial.  The rules are those
listed in Figure~\ref{fig:oml-eqn}, together with rules for reflexivity, symmetry, and transitivity
of equality, and the expected $\alpha$-equivalence and congruence rules for each syntactic form.
Rules \erule{$\beta$}, \erule{$\eta$}, \erule{$\mu$} and \erule{Let} should be unsurprising.  Rules
\erule{\I\then} and \erule{\E\then} mirror the corresponding typing rules, assuring that we can only
conclude equalities about well-typed expressions.  Rule \erule{\Eforall} should also be
unsurprising: if we have proved that two expressions are equal at a quantified type, we have that
they are equal at any of its instances.  Rule \erule{\Iforall} is less typical, as it requires one
subproof for each possible ground type ($GType$ ranges over ground type expressions).  Note that
this is only non-trivial for terms involving overloading.  Finally, rule \erule{Method} provides
(one step of) method resolution.  Intuitively, it says that for some class method $x$ at type
$\sigma$, if instance $d$ proves that $x$ is defined at $\sigma$, then $x$ is equal to the
implementation of $x$ provided by instance $d$.

\subsection{Polymorphic Identity Functions}\label{sec:poly-id-eqn}

In the introduction, we gave an example of a polymorphic function (!invBool!) that could be
instantiated at all types, yet did not have parametric behavior.  In this section, we will consider
a function which does have parametric behavior, but is defined in an ad-hoc fashion.  We will
demonstrate that our treatment of equality allows us to conclude that it is equal to its parametric
equivalent.

Our particular example is the identity function.  First, we give its typical definition:
\begin{code}
id1 :: t -> t
id1 x = x
\end{code}
For our second approach, we intend an overloaded definition that is provably equal to the parametric
definition.  We could produce such a definition using instance chains:
\begin{code}
class Id2' t where
  id2' :: t -> t
instance (Id2' t, Id2' u) => Id2' (t -> u) where
  id2' f = id2' . f . id2'
else Id2' t where
  id2' x = x
\end{code}
This gives an ad-hoc definition of the identity function, defined at all types but defined
differently for function and non-function types. Reasoning about this definition would require
extending the entailment relation to instance chains, introducing significant additional complexity.
We present simpler instances, but restrict the domain of types to achieve a similar result.
\begin{code}
class Id2 t where
  id2 :: t -> t
instance Id2 Int where
  id2 x = x
instance (Id2 t, Id2 u) => Id2 (t -> u) where
  id2 f = id2 . f . id2
\end{code}
We will use !Int! to stand in for all base (non-function) types.

It should be intuitive that, while they are defined differently, !id1 x! and !id2 x! should each
evaluate to !x! for any integer or function on integers !x!.  Correspondingly, given a class context
$\Psi$ that describes (at least) !Id2!, we can prove that \eqx{\mathtt{id1}}{\mathtt{id2}}{\tau} (we
omit the empty context and empty assumptions) for any such type $\tau$.  The case for integers is
direct: one application of \erule{Method} is sufficient to prove $\eqx{\mathtt{id2}}{\lambda
  x.x}{\mathtt{Int} \to \mathtt{Int}}$.  For functions of (functions of$\dots$) integers, the proof
has more steps, but is no more complicated.  For the simplest example, to show that
\[
 \eqx{\mathtt{id2}}{\lambda x.x}{(\mathtt{Int} \to \mathtt{Int}) \to (\mathtt{Int} \to \mathtt{Int})},
\]
we use \erule{Method} to show
\[
  \eqx{\mathtt{id2}}{\lambda f. (\mathtt{id2} \circ f \circ \mathtt{id2})}{(\mathtt{Int} \to \mathtt{Int}) \to (\mathtt{Int} \to \mathtt{Int})}.
\]
Relying on the usual definition of composition and \erule{$\beta$}, we show
\begin{multline*}
  \vdash_\Psi \lambda f. (\mathtt{id2} \circ f \circ \mathtt{id2}) \equiv
    \lambda f. \lambda x. \mathtt{id2} (f (\mathtt{id2}\,x)): \\
    (\mathtt{Int} \to \mathtt{Int}) \to (\mathtt{Int} \to \mathtt{Int})
\end{multline*}
Finally, by two uses of \erule{Method} for !id2! on integers, and \erule{$\eta$}, we have
\[
  \eqx{\lambda f. \lambda x. \mathtt{id2} (f (\mathtt{id2}\,x))}
      {\lambda f. f}
      {(\mathtt{Int} \to \mathtt{Int}) \to (\mathtt{Int} \to \mathtt{Int})}
\]
and thus the desired result.

We cannot expect to prove that $\mathtt{id1} \equiv \mathtt{id2}$ at all types (i.e.,
$\eqx{\mathtt{id1}}{\mathtt{id2}}{\forall t. t \to t}$) without limiting the domain of types.  For
example, there is no instance of !Id2! at type !Bool!; therefore, we cannot prove any non-trivial
equalities $\eqx{\mathtt{id2}}{M}{\mathtt{Bool} \to \mathtt{Bool}}$.  However, if we were to
restrict the grammar of types to those types for which !Id2! is defined (that is, if we define that
$\tau ::= \mathtt{Int} \mid \tau \to \tau$), then we could construct such an argument.  To show that
\(
  \eqx{\mathtt{id2}}{\lambda x.x}{\forall t. t \to t},
\)
we begin by applying \erule{\Eforall}, requiring a derivation
\(
  \eqx{\mathtt{id2}}{\lambda x.x}{\tau \to \tau}
\)
for each ground type $\tau$.  We could construct such a set of derivations by induction on the
structure of types, using the argument for !Int! above as the base case, and a construction
following the one for $\mathtt{Int} \to \mathtt{Int}$ for the inductive case.

A similar approach applies to the formulation using instance chains (class !Id2'!): we could show
that the first clause applied to functions, the second clause applied to any non-function type, and
use induction over the structure of types with those cases.

\section{A Simple Semantics for Overloading}\label{sec:sem-over}

Next, we develop a simple denotational semantics of \oml{} programs, extending an approach
originally proposed by Ohori~\cite{Ohori89} to describe the implicit polymorphism of ML.  As with
the presentation of equality in the previous section, the primary new challenges arise from the
definition of class methods and the treatment of overloading.  We will demonstrate that the
specialization-based approach is well-suited to addressing both challenges.  In particular, it
allows expressions to have different interpretations at each ground type without introducing
additional arguments or otherwise distinguishing qualified from unqualified type schemes.

\subsection{The Meaning of Qualified Types}\label{sec:meaning-qualified-types}

To describe the meaning of overloaded expressions, we must begin with the meaning of qualified
types.  Intuitively, qualifiers in types can be viewed as predicates in set comprehensions---that
is, a class !Eq! denotes a set of types, and the qualified type
\(
  \forall t. \predh{\mathtt{Eq}}{t} \then t \to t \to \mathtt{Bool}
\)
describes the set of types
\(
  \{ t \to t \to \mathtt{Bool} \mid t \in \mathtt{Eq} \}.
\)
However, most existing approaches to the semantics of overloading do not interpret qualifiers in
this fashion: Wadler and Blott~\cite{WadlerBlott89}, for instance, translate qualifiers into
dictionary arguments, while Jones~\cite{Jones92} translates qualified types into a calculus with
explicit evidence abstraction and application.

Our approach, by contrast, preserves the intuitive notion of qualifiers.  Given some class context
$\Psi = \tuple{A,Si,Im}$, we define the ground instances $\gr{\sigma}_\Psi$ of an \oml{} type scheme
$\sigma$ by:
\begin{align*}
  \gr\tau_\Psi &= \{ \tau \} \\
  \gr{\pi \then \rho}_\Psi &=
    \begin{cases}
      \gr\rho_\Psi &\text{if $\entails_A \pi$} \\
      \emptyset &\text{otherwise}
    \end{cases} \\
  \gr{\forall t. \sigma}_\Psi &= \bigcup_{\tau \in GType} \gr{[\tau/t]\sigma}_\Psi.
\end{align*}
Equivalently, if we define $GSubst(\seq t)$ to be substitutions that map $t$ to ground types and are
otherwise the identity, we have
\[
  \gr{\forall \seq{t}. P \then \tau}_\Psi = \{ S\,\tau \mid S \in GSubst(\seq t, \entails_A S\,P \}.
\]
We will omit annotation $\Psi$ when it is unambiguous.

In the typing judgments for \oml{}, predicates can appear in both types and contexts.  To account
for both sources of predicates, we adopt Jones's constrained type schemes $(P \mid \sigma)$, where
$P$ is a list of predicates and $\sigma$ is an \oml{} type scheme; an unconstrained type scheme
$\sigma$ can be treated as the constrained scheme $(\emptyset \mid \sigma)$ (as an empty set of
predicates places no restrictions on the instantiation of the variables in $\sigma$).  We can
define the ground instances of constrained type schemes by a straightforward extension of the
definition for unconstrained schemes: if $\Psi = \tuple{A,Si,Im}$, then
\[
  \gr{(P \mid \forall \seq{t}. Q \then \tau)}_\Psi = \{ S\,\tau \mid S \in GSubst(\seq t),\entails_A (P, S\,Q) \}.
\]

\subsection{Type Frames for Polymorphism}

We intend to give a semantics for \oml{} expressions by giving a mapping from their typing
derivations to type-indexed collections of monomorphic behavior.  We begin by fixing a suitable
domain for the monomorphic behaviors.  Ohori assumed an underlying type-frame semantics; his
translations, then, were from implicitly polymorphic terms to the interpretations of terms in the
simply-typed $\lambda$-calculus.  Unfortunately, we cannot apply his approach without some
extension, as type classes are sufficient to encode polymorphic recursion.  However, we can adopt
Harrison's extension~\cite{Harrison05} of Ohori's approach, originally proposed to capture
polymorphic recursion, and thus also sufficient for type class methods.

We begin by defining \emph{PCPO frames}, an extension of the standard notion of type frames.  A PCPO
frame is a tuple
\[
  \T = \tuple{\intype\cdot,\interm\cdot,T_{\tau,\upsilon},\sqsubseteq_\tau,\sqcup_\tau,\bot_\tau},
\]
(where we will omit the \textsf{type} and \textsf{term} annotations when they are apparent from
context) subject to the following six conditions.
\begin{enumerate}
\item For each ground type $\tau$, $\intype{\tau}$ is a non-empty set providing the interpretation
  of $\tau$.
\item For each typing derivation $\De$ of $\Gamma \vdash M : \tau$ and $\Gamma$-compatible
  environment $\eta$, $\interm{\De}\eta$ is the interpretation of $M$ in $\intype\tau$.
\item $T_{\tau,\upsilon} : \intype{\tau \to \upsilon} \times \intype{\tau} \to \intype{\upsilon}$
  provides the interpretation of the application of an element of $\tau \to \upsilon$ to an element
  of $\tau$.
\item For any $f,g \in \intype{\tau \to \upsilon}$, if, for all $x \in \intype\tau$,
  $T_{\tau,\upsilon}(f, x) = T_{\tau,\upsilon}(g, x)$, then $f = g$.
\item $\interm\cdot$ and $T_{\tau,\upsilon}$ respect the semantics of the simply-typed
  $\lambda$-calculus.  In particular:
  \begin{itemize}
  \item If $\De$ derives $\Gamma \vdash x : \tau$, then $\interp\De\eta = \eta(x)$;
  \item If $\De$ derives $\Gamma \vdash M\,N : \upsilon$, $\De_M$ derives $\Gamma \vdash M : \tau
    \to \upsilon$ and $\De_N$ derives $\Gamma \vdash N : \tau$, then $\interp\De\eta =
    T_{\tau,\upsilon}(\interp{\De_M}\eta,\interp{\De_N}\eta)$; and,
  \item If $\De_\lambda$ derives $\Gamma \vdash \lambda x: \tau.M : \tau \to \upsilon$
    and $\De_M$ derives $\Gamma,x:\tau \vdash M: \upsilon$, then
    $T_{\tau,\upsilon}(\interp{\De_\lambda}\eta, d) = \interp{\De_M}(\eta[x \mapsto d])$.
  \end{itemize}
\item Each set $\interp\tau$ is a PCPO with respect to $\sqsubseteq_\tau$, $\sqcup_\tau$ and
  $\bot_\tau$.
\end{enumerate}
The first five conditions are the standard requirements for type frames; the final condition relates
the type frame and PCPO structures of a PCPO frame.  Given a PCPO frame $\T$, we can define the
interpretation of a polymorphic type scheme $\sigma$ as the mappings from the ground instances
$\tau$ of $\sigma$ to elements of $\interp\tau$.  That is:
\[
  \intsch{\sigma}_\Psi = \Pi (\tau \in \gr\sigma_\Psi). \intype\tau.
\]
where we will omit the \textsf{scheme} and $\Psi$ annotations when it is not ambiguous.  For
example, the identity function $\lambda x. x$ has the type scheme $\forall t. t \to t$.  Therefore,
the semantics of the identity function is a map from the ground instances of its type (i.e., the
types $\tau \to \tau$) to the semantics of the simply-typed identity function at each type.  We
would expect its semantics to include the pair
\[
   \tuple{\mathtt{Int \to Int}, \interm{\vdash \lambda x : \mathtt{Int}. x: \mathtt{Int \to Int}}}
\]
to account for the $\mathtt{Int \to Int}$ ground instance of its type scheme, the pair
\[
   \tuple{\mathtt{Bool \to Bool}, \interm{\vdash \lambda x : \mathtt{Bool}. x: \mathtt{Bool \to Bool}}}
\]
to account for the $\mathtt{Bool \to Bool}$ ground instance of its type scheme, and so forth.  Note
that if $\sigma$ has no quantifiers, and so $\gr\sigma_\Psi = \{ \tau \}$ for some type $\tau$, then we
have
\[
  \intsch\sigma_\Psi = \{ \{ \tuple{\tau, b} \} \mid b \in \intype\tau \},
\]
and so an element of $\intsch{\tau}$ is a singleton map, not an element of $\intype\tau$.  Harrison
proves that $\interp\sigma$ is itself a pointed CPO, justifying solving recursive equations in
$\interp\sigma$.
\begin{theorem}[Harrison]
  Let $\T$ be a PCPO frame.  Then, for any type scheme $\sigma$, $\interp\sigma$ is a pointed CPO
  where:
  \begin{itemize}
  \item For any $f,g \in \interp\sigma$, $f \sqsubseteq_\sigma g \iff (\forall \tau \in
    \gr\sigma. \: f(\tau) \sqsubseteq_\tau g(\tau))$;
  \item The bottom element $\bot_\sigma$ is defined to be $\{ \tuple{\tau, \bot_\tau} \mid \tau \in
    \gr\sigma \}$; and,
  \item The least upper bound of an ascending chain $\{ f_i \} \subseteq \interp\sigma$ is
    $\{\tuple{\tau,u_\tau} \mid \tau \in \gr\sigma, u_\tau = \sqcup_\tau (f_i(\tau))\}$.
  \end{itemize}
\end{theorem}
\noindent We can define continuous functions and least fixed points for sets $\interp\sigma$ in the
usual fashion:
\begin{itemize}
\item A function $f : \interp\sigma \to \interp{\sigma'}$ is continuous if $f(\sqcup_\sigma X_i) =
  \sqcup_{\sigma'}(f(X_i))$ for all directed chains $X_i$ in $\interp\sigma$.
\item The fixed point of a continuous function $f : \interp\sigma \to \interp\sigma$ is defined by
  $\fix(f) = \sqcup_\sigma (f^n(\bot_\sigma))$, and is the least value such that $\fix(f) =
  f(\fix(f))$.
\end{itemize}

\subsection{Semantics for Overloaded Expressions}

We can now give denotations for (typing derivations of) \oml{} expressions.  For some type
environment $\Gamma$ and substitution $S \in GSubst(ftv(\Gamma))$, we define an
$S-\Gamma$-environment $\eta$ as a mapping from variables to values such that $\eta(x) \in
\interp{(S\,\sigma)}$ for each assignment $(x : \sigma)$ in $\Gamma$. Given a PCPO frame $\T$, a
derivation $\De$ of $\tys{M}{\sigma}$, a ground substitution $S$, and an environment $\eta$, we
define the interpretation $\interp{\De}S\eta$ by cases.  We have included only a few, representative
cases here.

\begin{itemize}
\item Case \trule{\E\to}: we have a derivation of the form
  \[
    \infbox{\lproof[\De]
                   {\lproof[\De_1]
                           {\vdots};
                           {\tys{M}{\tau \to \tau'}}}
                   {\lproof[\De_2]
                           {\vdots};
                           {\tys{N}{\tau}}};
                   {\tys{(M\,N)}{\tau'}}}
  \]
  Let $\upsilon = S\,\tau$ and $\upsilon' = S\,\tau'$, and define
  \begin{multline*}
    \interp{\De}S\eta = \{ \tuple{ \upsilon',
      T_{\upsilon,\upsilon'}((\interp{\De_1}S\eta)(\upsilon \to \upsilon'), \\ (\interp{\De_2}S\eta)(\upsilon)) } \}.
  \end{multline*}
\item Case \trule{\I\then}: we have a derivation of the form
  \[
    \infbox{\lproof[\De]
                   {\lproof[\De_1]
                           {\vdots};
                           {\typ{P,\pi}{\Gamma}{A}{M}{\rho}}};
                   {\tys{M}{\pi \then \rho}}}
  \]
  This rule excludes those cases in which the predicate does not hold; thus, we define:
  \[
    \interp{\De}S\eta = \begin{cases}
      \interp{\De_1}S\eta &\text{if $S\,P \entails S\,\pi$;} \\
      \emptyset &\text{otherwise.}
    \end{cases}
  \]
\item Case \trule{\E\then}: we have a derivation of the form
  \[
    \infbox{\lproof[\De]
                   {\lproof[\De_1]
                           {\vdots};
                           {\tys{M}{\pi \then \rho}}}
                   {\assume{P \entails \pi}};
                   {\tys{M}{\rho}}}
  \]
  This rule does not affect the semantics of expression $M$, and so we define:
  \[
    \interp{\De}S\eta = \interp{\De_1}S\eta.
  \]
\item Case \trule{\Iforall}: we have a derivation of the form
  \[
    \infbox{\lproof[\De]
                   {\lproof[\De_1]
                           {\vdots};
                           {\tys{M}{\sigma}}}
                   {\assume{t \not\in ftv(P,\Gamma)}};
                   {\tys{M}{\forall t. \sigma}}}
  \]
  Intuitively, we interpret a polymorphic expression as the map from ground instances of its type to
  its interpretations at those types.  As the interpretation of the subderivation $\De_1$ is
  already in the form of a such a map, we can interpret $\De$ as the union of the meanings of
  $\De_1$ for each ground instantiation of the quantified variable $t$.  Formally, we define
  \[
    \interp{\De}S\eta = \bigcup_{\tau \in GType} \interp{\De_1}(S[t \mapsto \tau])\eta.
  \]
\item Case \trule{\Eforall}: we have a derivation of the form
  \[
    \infbox{\lproof[\De]
                   {\lproof[\De_1]
                           {\vdots};
                           {\tys{M}{\forall t. \sigma}}};
                   {\tys{M}{[\tau/t]\sigma}}}
  \]
  By definition, $\gr{\forall t. \sigma} = \bigcup_{\tau \in GType}\gr{[\tau/t]\sigma}$, and so
  $\gr{[\tau/t]\sigma} \subseteq \gr{\forall t. \sigma}$.  Thus, the interpretation of $\De$ is a
  subset of the interpretation of $\De_1$; writing $\restrict{f}{Y}$ for the restriction of a
  function $f$ to some subset $Y$ of its domain, we define:
  \[
    \interp{\De}S\eta = \restrict{(\interp{\De_1}S\eta)}{\gr{[\tau/t]\sigma}}.
  \]
\end{itemize}

\subsection{Expressions with Class Contexts}\label{sec:ctxt-sem}

To complete our semantics of \oml{} programs, we must account for the meaning of class methods.  Our
approach is intuitively simple: we collect the meanings of the class methods from the method
implementations in each instance, and use the meanings of the methods to define the meaning of the
main expression. Formally, we extend the interpretation function from derivations of
$\tys{M}{\sigma}$ to derivations of $\typ{P}{\Gamma}{\Psi}{M}{\sigma}$ as follows:
\begin{itemize}
\item Let $\De$ be a derivation of $\typ{P}{\Gamma}{\Psi}{M}{\tau}$.  Then we know that $\De$
  must begin with an application of \trule{Ctxt} (Figure~\ref{fig:prog-typing-oml}) with one
  subderivation
  \[
    \infbox{\lproof[\De_{y,d}]
                   {\vdots};
                   {\typ{P}{\Gamma,\overline{x_i: \sigma_{x_i}}}{A}{Im(y,d)}{\sigma_{y,d}}}}
  \]
  for each pair $\tuple{y,d} \in \dom(Im)$ and a subderivation
  \[
    \infbox{\lproof[\De_M]
                   {\vdots};
                   {\typ{P}{\Gamma, \overline{x_i: \sigma_{x_i}}}{A}{M}{\tau}}}
  \]
  for the main expression $M$.  We enumerate the methods in the program as $x_1,x_2,\dots,x_m$, and
  let
  \[
    \Sigma = \interp{\sigma_{x_1}} \times \interp{\sigma_{x_2}} \times \cdots \times \interp{\sigma_{x_m}}.
  \]
  For each method $x_i$, we define a function $f_i : \Sigma \to \interp{\sigma_{x_i}}$, approximating its
  meaning, as follows:
  \[
    f_i(\tuple{b_1, b_2, \dots,b_m})S\eta = \bigcup_{\tuple{x_i,d} \in \dom(Im)} \interp{\De_{x_i,d}}S(\eta[\overline{x_j} \mapsto \overline{b_j}]),
  \]
  and define function $f : \Sigma \to \Sigma$, approximating the meaning of all the methods in the
  program, as
  \[
    f(b) = \tuple{f_1(b), f_2(b), \dots, f_m(b)}.
  \]
  We can now define a tuple $b$, such that the component $b_i$ is the meaning of method $x_i$, as
  follows:
  \[
    b = \sideset{}{_\Sigma}\bigsqcup f^n(\bot_\Sigma).
  \]
  Finally, we extend the interpretation function to programs by
  \[
    \interp{\De}S\eta = \interp{\De_M}S(\eta[\overline{x_i} \mapsto \overline{b_i}]).
  \]
\end{itemize}

\subsection{Polymorphic Identity Functions Revisited}\label{sec:poly-id-den}

We return to our earlier example of polymorphic identity functions~\secref{poly-id-eqn}.  As before,
we consider two definitions of identity functions, one given parametrically (!id1!) and one given by
overloading (!id2!).  In this section, we will show that the denotations of !id1! and !id2! agree at
all types for which !id2! is defined.  By doing so, we provide an intuitive demonstration that our
denotational semantics captures the meaning of ad-hoc polymorphic and agrees with our definition of
equality for \oml{} terms.

\def\id2h#1{\predh{\mathtt{Id2}}{#1}}
\def\semone{\interp{\mathtt{id1}}}
\def\semtwo{\interp{\mathtt{id2}}}

We show that $\semone$ and $\semtwo$ have the same value at each point in the domain of $\semtwo$; that
is, that for any type $\tau \in GType$ such that $\entails \predh{\mathtt{Id2}}{\tau}$,
\[
  \semone(\tau \to \tau) = \semtwo(\tau \to \tau).
\]
We proceed by induction on the structure of $\tau$.  In the base case, we know that $\tau = K$ for
some non-functional type $K$.  As we have assumed $\entails \id2h\tau$, we must have that $K =
\mathtt{Int}$, and, from the instances for !Id2!, we have
\begin{align*}
  \semtwo(K \to K) &= \semtwo(\mathtt{Int} \to \mathtt{Int}) \\
   &= \interp{\vdash \lambda x:\mathtt{Int}. x:\mathtt{Int \to Int}}.
\end{align*}
As
\(
  \semone(\mathtt{Int} \to \mathtt{Int}) = \interp{\vdash \lambda x:\mathtt{Int}. x:\mathtt{Int \to Int}},
\)
we have
\(
  \interp{\mathtt{id1}}(K \to K) = \semtwo(K \to K).
\)
In the inductive case, we know that $\tau = \tau_0 \to \tau_1$ for some types $\tau_0$ and $\tau_1$.
From the assumption that $\entails \id2h{(\tau_0 \to \tau_1)}$ and the instances for !Id2!, we can
assume that $\id2h\tau_0$, $\id2h\tau_1$, and that
\[
  \semtwo(\tau \to \tau) =
    \interp{\vdash \lambda f : (\tau_0 \to \tau_1). M \circ f \circ N : \tau \to \tau}
\]
for some simply typed expressions $M$ and $N$ such that $\interp{M} = \semtwo(\tau_0 \to \tau_0)$
and $\interp{N} = \semtwo(\tau_1 \to \tau_1)$.  The induction hypothesis gives that $\semtwo(\tau_0
\to \tau_0) = \semone(\tau_0 \to \tau_0)$ and that $\semtwo(\tau_1 \to \tau_1) = \semone(\tau_1 \to
\tau_1)$, and thus that $\interp{M} = \interp{\vdash \lambda x: \tau_1. x: \tau_1 \to \tau_1}$ and
$\interp{N} = \interp{\vdash \lambda x: \tau_0. x: \tau_0 \to \tau_0}$.  By congruence, we have
\[
  \semtwo(\tau \to \tau) =
    \interp{\lambda f : (\tau_0 \to \tau_1). (\lambda x : \tau_1. x) \circ f \circ (\lambda x : \tau_0 . x)}.
\]
Finally, assuming a standard definition of composition, and reducing, we have
\begin{align*}
  \semtwo(\tau \to \tau)
  & = \interp{\lambda f: (\tau_0 \to \tau_1). f} \\
  & = \interp{\lambda f: \tau. f} \\
  & = \semone(\tau \to \tau).
\end{align*}

In our previous discussion of this example, we argued that if the set of types were restricted to
those types for which !Id2! held, then !id1! and !id2! were equal.  We can show a similar result
here, by showing that if we define that $\tau ::= \mathtt{Int} \mid \tau \to \tau$, then $\semone =
\semtwo$.  We begin by showing that they are defined over the same domain; that is, that \(
\gr{\forall t. \: t \to t} = \gr{\forall u. \id2h{u} \then u \to u}.  \) By definition, we have
\[
  \gr{\forall t. \: t \to t} = \{ \tau \to \tau \mid \tau \in GType \}
\]
and
\[
  \gr{\forall u. \: \id2h{u} \then u \to u} =
    \{ \tau \to \tau \mid \tau \in GType, \entails \id2h\tau \}.
\]
We show that $\entails \predh{\mathtt{Id2}}{\tau}$ for all types $\tau$ by induction on the
structure of $\tau$.  In the base case, we know that $\tau = \mathtt{Int}$, and by the first
instance of !Id2! we have $\entails \id2h\tau$.  In the inductive case, we know that $\tau = \tau_0
\to \tau_1$ for some types $\tau_0,\tau_1$.  In this case, we have that $[\tau_0/t, \tau_1/u] \tau =
t \to u$ and by the induction hypothesis, that $\entails \id2h\tau_0$ and $\entails \id2h\tau_1$.
Thus, from the second instance of !Id2!, we can conclude that $\entails \id2h{(\tau_o \to \tau_1)}$,
that is, that $\entails \id2h\tau$.  Because $\entails \id2h\tau$ for all ground types $\tau$, we
have
\[
  \{ \tau \to \tau \mid \tau \in GType, \entails \id2h\tau \} =
    \{ \tau \to \tau \mid \tau \in GType \},
\]
and so $\semone$ and $\semtwo$ are defined over the same domain.  We have already shown that
$\semone$ and $\semtwo$ agree at all points at which they are defined, and so we conclude $\semone =
\semtwo$.

\section{Formal Properties}\label{sec:formal}

The previous sections have outlined typing and equality judgments for \oml{} terms, and proposed a
denotational semantics for \oml{} typings.  In this section, we will relate these two views of the
language.  We begin by showing that the denotation of a typing judgment falls into the expected
type.  This is mostly unsurprising; the only unusual aspect of \oml{} in this respect is the role of
the class context.  We go on to show that the equational judgments are sound; again, the unusual
aspect is to do with polymorphism (\erule{\Iforall} and \erule{\Eforall}) and class methods
(\erule{Method}).  The \oml{} type system follows Jones's original formulation of OML; we rely on
several of his metatheoretical results, such as the closure of typing under substitution.

\begin{theorem}[Soundness of typing]\label{thm:type-soundness}
  Given a class context $\Psi$, if $\De$ is a derivation of $\typ{P}{\Gamma}{\Psi}{M}{\sigma}$, $S$
  is a substitution, and $\eta$ is an $(S\,\Gamma)$-environment, then $\interp{\De}S\eta
  \in \interp{(S\,P \mid S\,\sigma)}_\Psi$.
\end{theorem}

We will divide the proof into three pieces. First, we show the soundness of the judgment
$\tys{M}\sigma$.  Then, we will argue that the union of the implementations of a method has the type
of the method itself.  Finally, we can combine these results to argue the soundness of
$\typ{P}{\Gamma}{\Psi}{M}{\sigma}$.

\begin{lemma}\label{thm:expr-type-soundness}
  Given a class context $\Psi = \tuple{A,Si,Im}$ where $A$ is non-overlapping, if $\De$ is a
  derivation of $\tys{M}{\sigma}$, $S$ is a substitution, and $\eta$ is a
  $(S\,\Gamma)$-environment, then $\interp{\De}S\eta \in \interp{(S\,P \mid S\,\sigma)}_\Psi$.
\end{lemma}

\begin{proof}
  The proof is by induction over the structure of derivation $\De$.  The cases are straightforward;
  we include several representative examples.  (Meta-variables $\De_n$ are as in the definition of
  $\interp\cdot$ above.)
  \begin{itemize}
  \item \emph{Case \trule{$\then$ I}.} Observe that $\gr{(S(P, \pi) \mid S\,\rho)} = \gr{(S\,P \mid
    S\,(\pi \then \rho))}$. As such, if
    \[
      \interp{\De_1}S\eta \in \interp{(S\,(P, \pi) \mid S\,\rho)}_\Psi,
    \]
    then we must also have that
    \[
      \interp{\De}S\eta \in \interp{(S\,P \mid S\,(\pi \then \rho))}_\Psi.
    \]
  \item \emph{Case \trule{$\then$ E}.}  As entailment is (trivially) closed under substitution, $P
    \entails \pi$ implies that $S\,P \entails S\,\pi$ for any substitution $S$; thus, we can conclude
    that $\gr{(S\,P \mid S\,(\pi \then \rho))} = \gr{(S\,P \mid S\,\rho)}$.  Finally, assuming that
    $\interp{\De_1}S\eta \in \interp{(S\,P \mid S\,(\pi \then \rho))}$, we can conclude that
    $\interp{\De}S\eta \in \interp{(S\,P \mid S\,\rho)}$.
  \item \emph{Case \trule{\Iforall}.}  Because $\sigma = \forall t. \sigma'$, we have that
    \[\gr\sigma = \bigcup_{\tau \in GType}\gr{[\tau/t]\sigma'},\] and thus that \[\interp\sigma =
    \bigcup_{\tau \in GType} (\interp{[\tau/t]\sigma'}).\] Thus, assuming that for ground types
    $\tau$, $\interp{\De_1}(S[t \mapsto \tau])\eta \in \interp{(S\,P \mid S\,\sigma')}$, we have
    \[
      \interp{\De}S\eta \in \left(\bigcup_{\tau \in GType} \interp{(S\,P \mid S\,\sigma')}\right) = \interp{(S\,P \mid S\,\sigma)}.
    \]
  \item \emph{Case \trule{\Eforall}.}  Assuming that $\interp{\De_1}S\eta \in \interp{(S\,P \mid
    S\,(\forall t. \sigma'))}$, the same argument about ground types as in the previous case gives
    that $\interp{\De}S\eta \in \interp{(S\,P \mid S\,\sigma)}$. \qedhere
  \end{itemize}
\end{proof}

The interpretation of typings $\typ{P}{\Gamma}{\Psi}{M}{\sigma}$ depends on the interpretations of
the class methods.  We will begin by showing that the interpretation of each method is in the
denotation of its type.  To do so, we will demonstrate that the interpretation of the type scheme of
a method is the union of the interpretation of the type schemes of its instances.  This will show
that the union of the implementations is in the type of the method, from which the desired result
follows immediately.

\begin{lemma}\label{thm:method-type-schemes}
  The ground instances of the type scheme of a method $x$ are the union of its ground instances at
  each of its instances.  That is,
  \[
    \gr{\sigma_x} = \bigcup_{\tuple{x,d} \in \dom(Im)} \gr{\sigma_{x,d}}.
  \]
\end{lemma}

\begin{proof}
  Let $\sigma_x = \forall \seq{t}. (\pi, Q) \then \tau$, where $x$ is a method of $class(\pi)$.  We
  prove that
  \[
    \gr{\sigma_x} = \bigcup_{\tuple{d,x} \in \dom(Im)} \gr{\sigma_{x,d}}
  \]
  by the inclusions
  \[
    \gr{\sigma_x} \subseteq \bigcup_{\tuple{x,d} \in \dom(Im)} \gr{\sigma_{x,d}},
  \]
  and
  \[
    \gr{\sigma_x} \supseteq \bigcup_{\tuple{x,d} \in \dom(Im)} \gr{\sigma_{x,d}}.
  \]
  We will show only the first inclusion; the second is by an identical argument.  Fix some $\upsilon
  \in \gr{\sigma_x}$.  By definition, there is some $S \in GSubst(\seq t)$ such that $\upsilon =
  S\,\tau$ and $\entails S\,\pi,S\,Q$.  Because $\entails S\,\pi$, there must be some
  $(\qclause{d}{\seq{u}}{P}{\pi'}) \in A$ and substitution $S' \in GSubst(\seq u)$ such that
  $S\,\pi = S'\,\pi'$ and $\entails S'\,P$.  Now, we have that $\sigma_{x,d} = \forall \seq{t}'. (P,
  T\,Q) \then T\,\tau$ for some substitution $T$; thus, there is some $T' \in GSubst(\seq t')$ such
  that $\upsilon = T'\,(T\,\tau)$, $S\,P = T'\,(T\,Q)$, and so $\upsilon \in \gr{\sigma_{x,d}}$.
\end{proof}

\begin{lemma}\label{thm:method-types}
  The interpretation of the type scheme of a method $x$ is the union of the interpretations of its
  type scheme at each instance.  That is,
  \[\interp{\sigma_x} = \bigcup_{\tuple{x,d} \in \dom(Im)} \interp{\sigma_{x,d}}.\]
\end{lemma}

\begin{proof}
  Recall that
  \[
    \intsch{\sigma_x} = \Pi (\tau \in \gr{\sigma_x}). \intype{\tau}.
  \]
  From Lemma~\ref{thm:method-type-schemes}, we have that
  \[
    \intsch{\sigma_x} = \Pi \left(\tau \in \bigcup_{\tuple{x,d} \in \dom(Im)} \gr{\sigma_{x,d}} \right).
      \intype{\tau}.
  \]
  As $\intype\cdot$ is a function, this is equivalent to
  \[
    \intsch{\sigma_x} = \bigcup_{\tuple{x,d} \in \dom(Im)} \Pi(\tau \in \gr{\sigma_{x,d}}). \intype\tau,
  \]
  and finally, again from the definition of $\intsch\cdot$,
  \[
    \intsch{\sigma_x} = \bigcup_{\tuple{x,d} \in \dom(Im)} \intsch{\sigma_{x,d}}. \qedhere
  \]
\end{proof}

\begin{proof}[Proof of Theorem~\ref{thm:type-soundness}]
  Finally, we can extend the soundness of our semantics to include class contexts.  From
  Lemmas~\ref{thm:method-type-schemes} and~\ref{thm:method-types}, we know that the interpretations
  of the methods fall in the interpretations of their type schemes, and so if $\eta$ is a
  $S-\Gamma$-environment, then $\eta[\overline{x_i} \mapsto \overline{b_i}]$ is a
  $S-(\Gamma,\overline{x_i:\sigma_{x_i}})$-environment.  From Theorem~\ref{thm:expr-type-soundness}, we have that
  $\interp{\De_M}S(\eta[\overline{x_i} \mapsto \overline{b_i}) \in \interp{(S\,P \mid
      S\,\sigma)}_\Psi$, and thus that $\interp{\De}S\eta \in \interp{(S\,P \mid S\,\sigma)}_\Psi.$
\end{proof}

We would like to know that the meaning of an expression is independent of the particular choice of
typing derivation.  Unfortunately, this is not true in general for systems with type classes.  A
typical example involves the !read! and !show! methods, which have the following type signatures
\begin{code}
read :: Read t => String -> t
show :: Show t => t -> String
\end{code}
We can construct an expression !show . read! of type
\[
  (\predh{\mathtt{Read}}{t},\predh{\mathtt{Show}}{t}) \then \mathtt{String} \to \mathtt{String},
\]
where variable $t$ can be instantiated arbitrarily in the typing, changing the meaning of the
expression.  To avoid this problem, we adopt the notion of an unambiguous type scheme from Jones's
work on coherence for qualified types~\cite{Jones93Coherence}.

\begin{defn}\label{def:ambiguous}
  A type scheme $\sigma = \forall \vec t. P \then \tau$ is unambiguous if $ftv(P) \subseteq
  ftv(\tau)$.
\end{defn}

\noindent
As long as we restrict our attention to unambiguous type schemes, we have the expected coherence
result.  For example, suppose that $\De$ is a derivation of $\tys{\lambda x.M}{\sigma}$.  We observe
that $\De$ must conclude with an application of \trule{\I\to}, say at $\typ{P_0}{\Gamma}{A}{\lambda
  x.M}{\tau \to \tau'}$, followed by a series of applications of \trule{\I\then}, \trule{\E\then},
\trule{\Iforall} and \trule{\Eforall}.  While these latter applications determine $\sigma$, we can
see intuitively that each $\upsilon \in \gr\sigma$ must be a substitution instance of $\tau \to
\tau'$, and that the interpretation of $\Delta$ at each ground type must be the interpretation of an
instance of the subderivation ending with \trule{\I\to}.  We can formalize these two observations by
the following lemma.

\newcommand{\proveq}{\ensuremath{\dashv \vdash}}

\begin{lemma}\label{thm:quantifiers}
  If $\sigma = \forall \seq t. Q \then \tau$, and $\De_1\dots\De_n$ is a sequence of derivations
  such that:
  \begin{itemize}
  \item $\Delta_1$ is a derivation of $\typ{P_1}{\Gamma}{A}{M}{\tau_1}$;
  \item $\Delta_n$ is a derivation of $\tys{M}{\sigma}$;
  \item Each of $\Delta_2\dots\Delta_n$ is by \trule{\I\then}, \trule{\E\then}, \trule{\Iforall} or
    \trule{\Eforall}; and,
  \item Each $\Delta_i$ is the principal subderivation of $\Delta_{i+1}$
  \end{itemize}
  then
  \begin{enumerate}[(a)]
  \item There is a substitution $S$ such that $\tau = S\,\tau_1$ and $P \cup Q \proveq S\,P_1$; and,
  \item For all ground substitutions S, for all $\upsilon \in \gr{S\,\sigma}$, there is a unique
    $S'$ such that $\interp{\Delta_n}S\eta\upsilon = \interp{\Delta_1}S'\eta\upsilon$.
  \end{enumerate}
\end{lemma}

\noindent
The proof is by induction on $n$; the cases are all trivial.  We can now characterize the
relationship between different typings of $M$.

\begin{theorem}[Coherence of $\interp\cdot$]\label{thm:coherence}
  If $\De$ derives $\tys{M}{\sigma}$ and $\De'$ derives $\typ{P'}{\Gamma'}{A}{M}{\sigma'}$, where
  $\sigma$ and $\sigma'$ are unambiguous, then for all substitutions $S$ and $S'$ such that $S\,P
  \proveq S'\,P', S\,\Gamma = S'\,\Gamma'$, and $S\,\sigma = S'\sigma'$, and for all ground
  substitutions $U$, $\interp{\Delta}(U \circ S) = \interp{\Delta'}(U \circ S')$.
\end{theorem}

\noindent
The proof is by induction over the structure of $M$.  In each case, use of the inductive hypothesis
is justified by Lemma~\ref{thm:quantifiers}(a), and the conclusion derived from the definition of
$\interp\cdot$ and Lemma~\ref{thm:quantifiers}(b).  As an immediate corollary, we have that if $\De$
and $\De'$ are two derivations of the same typing judgment, then $\interp{\De} = \interp{\De'}$.  We
can also show that, if $\tys{M}{\sigma}$ is a principal typing of $M$, with derivation $\Delta$, and
$\Delta'$ derives $\tys{M}{\sigma'}$ for any other $\sigma'$, then for each substitution $S'$ there
is a unique $S$ such that, for all environments $\eta$, $\interp{\Delta}S\eta \supseteq
\interp{\Delta'}S'\eta$.

\begin{theorem}[Soundness of $\equiv$]\label{thm:equiv-sound}
  Given a class context $\Psi$, if $\sigma$ is unambiguous, $\eqs{M}{N}{\sigma}$, and $\De_M,\De_N$
  are derivations of $\typ{P}{\Gamma}{\Psi}{M}{\sigma},\typ{P}{\Gamma}{\Psi}{N}{\sigma},$ then
  $\interp{\De_M} = \interp{\De_N}$.
\end{theorem}

\begin{proof}
  The proof is by induction over the derivation of $\eqs{M}{N}{\sigma}$.  The interesting cases are
  to do with polymorphism and overloading.
  \begin{itemize}
  \item \textit{Case \erule{$\then$ I}.} We have a derivation concluding
    \[
       \infbox{\iproof{\assume{\eql{P,\pi}{\Gamma}{\Psi}{M}{N}{\rho}}};
                      {\eqs{M}{N}{\pi \then \rho}}}
    \]
    Let $\De_M,\De_N$ be typing derivations of $\tys{M}{\pi \then \rho}$ and $\tys{N}{\pi \then
      \rho}$; without loss of generality (because of Theorem~\ref{thm:coherence}), assume that each
    is by \trule{\I\then}, with subderivations $\De_M',\De_N'$ of
    $\typ{P,\pi}{\Gamma}{\Psi}{M}{\rho}$ and $\typ{P,\pi}{\Gamma}{\Psi}{N}{\rho}$.  From the
    definition of $\interp\cdot$, we have $\interp{\De_M} = \interp{\De_M'}$ and $\interp{\De_N} =
    \interp{\De_N'}$.  The induction hypothesis gives that $\interp{\De_M'} = \interp{\De_N'}$, and
    so we can conclude $\interp{\De_M} = \interp{\De_N}$.
  \item \textit{Case \erule{$\then$ E}.} We have a derivation concluding
    \[
      \infbox{\iproof{\assume{\eqs{M}{N}{\pi \then \rho}}}
                     {\assume{P \entails_A \pi}};
                     {\eqs{M}{N}{\rho}}}
    \]
    where $\Psi = \tuple{A,Si,Im}$.  As in the previous case, the interpretation of the typing
    derivations for $\typ{P}{\Gamma}{\Psi}{M}{\rho}$ and $\typ{P}{\Gamma}{\Psi}{M}{\pi \then \rho}$
    are equal, and similarly for the typing derivations for $N$, and thus the induction hypothesis
    is sufficient for the desired conclusion.
  \item \textit{Case \erule{\Iforall}.} We have a derivation concluding
    \[
      \infbox{\irule{\assume{\{ (\eqs{M}{N}{[\tau/t]\sigma}) \mid \tau \in GType \}}};
                    {\eqs{M}{N}{\forall t. \sigma}}}
    \]
    From the induction hypothesis, we can conclude that, given derivations $\De_M^\tau$ of
    $\typ{P}{\Gamma}{\Psi}{M}{[\tau/t]\sigma}$ and $\De_N^\tau$ of
    $\typ{P}{\Gamma}{\Psi}{N}{[\tau/t]\sigma}$, $\interp{\De_M^\tau} = \interp{\De_N^\tau}$.  Let
    $\De_M$ derive $\typ{P}{\Gamma}{\Psi}{M}{\forall t. \sigma}$ (and, without loss of generality,
    assume $\De_M$ is by \trule{\I\forall}); we know that $\interp{\De_M} = \bigcup_{\tau \in GType}
    \interp{\De_M^\tau}$.  We argue similarly for derivations $\De_N$ of
    $\typ{P}{\Gamma}{\Psi}{N}{\forall t. \sigma},$ and conclude that $\interp{\De_M} =
    \interp{\De_N}.$
  \item \textit{Case \erule{\Eforall}.} We have a derivation concluding
    \[
      \infbox{\irule{\assume{\eqs{M}{N}{\forall t. \sigma}}};
                    {\eqs{M}{N}{[\tau/t]\sigma}}}
    \]
    Let $\De_M,\De_N$ be derivations that $M$ and $N$ have type $[\tau/t]\sigma$; without loss of
    generality, assume they are by \trule{\E\forall}, with subderivations $\De_M',\De_N'$ that $M$
    and $N$ have type $\forall t. \sigma$.  From the induction hypothesis, we know $\interp{\De_M'}
    = \interp{\De_N'}$, and from the definition of $\interp\cdot$ we know that $\interp{\De_M}
    \subseteq \interp{\De_M'}$ and $\interp{\De_N} \subseteq \interp{\De_N'}$.  Thus, we can
    conclude that $\interp{\De_M} = \interp{\De_N}$.
  \item \textit{Case \erule{Method}.} We have a derivation of the form
    \[
      \infbox{\irule{Si(x) = \pi,\sigma}{d:P \entails_A S\,\pi};
                    {\eql{P}{\Gamma}{\tuple{A,Si,Im}}{x}{Im(x,d)}{S\,\sigma}}}
    \]
    Let $\De_M$ be the derivation of $\typ{P}{\Gamma}{\Psi}{x}{S\,\sigma}.$ From the definition of
    $\interp\cdot$, we know that $\interp{\De_M} S \eta = \interp{\De_M'} S (\eta[\overline{x_i}
      \mapsto \overline{b_i}])$ where the $x_i$ are the class methods, the $b_i$ are their
    implementations, and $\De_M'$ is the derivation of
    $\typ{P}{\Gamma,x_i:\sigma_i}{A}{x}{S\,\sigma}$.  Since $x$ is a class method, we know that
    $\eta[\overline{x_i} \mapsto \overline{b_i}]$ maps $x$ to some method implementation $b_j$, and
    therefore that $\interp{\De_M'} \subseteq b_j$.  We also know that $b_j$ is the fixed point of a
    function $f_j(\tuple{b_1,\dots,b_n})S\eta = \bigcup_d \interp{\De_{x,d'}}S(\eta[\overline{x_i}
      \mapsto \overline{b_i}])$, where $\De_{x,d'}$ derives $\typ{P}{\Gamma}{A}{Im(x,d')}{\sigma_{x,d'}}$ and
    $d$ is one of the $d_i$.  Thus, we know that if $\De_N$ derives
    $\typ{P}{\Gamma}{\Psi}{Im(x,d)}{S\,\sigma}$, then $\interp{\De_N} \subseteq b_j$.  Finally, as
    $\interp{\De_M}$ and $\interp{\De_N}$ are defined over the same domain, we have that
    $\interp{\De_M} = \interp{\De_N}$.  \qedhere
  \end{itemize}
\end{proof}

\section{Improvement and Functional Dependencies}\label{sec:fundeps}

In the introduction, we set out several ways in which extensions of type class systems went beyond
the expressiveness of existing semantic approaches to overloading.  In this section, we return to
one of those examples, demonstrating the flexibility of our specialization-based approach to
type-class semantics.

Functional dependencies~\cite{Jones00} are a widely-used extension of type classes which capture
relationships among parameters in multi-parameter type classes.  Earlier, we gave a class !Elems! to
abstract over common operations on collections:
\begin{code}
class Elems c e | c -> e where
  empty :: c
  insert :: e -> c -> c
\end{code}
The functional dependency $\mathtt{c} \to \mathtt{e}$ indicates that the type of a collection (!c!)
determines the type of its elements (!e!).  Practically speaking, this has two consequences:
\begin{itemize}
\item A program is only valid if the instances in the program respect the declared functional
  dependencies.  For example, if a program already contained an instance which interpreted lists as
  collections:
\begin{code}
instance Elems [t] t where ...
\end{code}
  the programmer could not later add an instance that interpreted strings (lists of characters in
  Haskell) as collections of codepoints (for simplicity represented as integers):
\begin{code}
instance Elems [Char] Int
\end{code}
\item Given two predicates $\Elems{\tau}{\upsilon}$ and $\Elems{\tau'}{\upsilon'}$, if we know $\tau
  = \tau'$, then we must have $\upsilon = \upsilon'$ for both predicates to be satisfiable.
\end{itemize}

We now consider an extension of \oml{} to support functional dependencies.  Following
Jones~\cite{Jones95}, we introduce a syntactic characterization of improving substitutions, one way
of describing predicate-induced type equivalence.  We then extend the typing and equality judgments
to take account of improving substitutions.  Finally, we show that the extended systems are sound
with respect to our semantics.  Importantly, we do not have to extend the models of terms, nor do we
introduce coercions, or other intermediate translations.  We need only show that our
characterization of improving substitutions is sound to show that the resulting type equivalences
hold in the semantics.

\subsection{Extending \omlhead{} with Functional Dependencies}
\newcommand{\impr}[3]{#1 \vdash #2\,\text{improves}\,#3}

To account for the satisfiability of predicates in qualified types, Jones introduces the notion of
an improving substitution $S$ for a set of predicates $P$~\cite{Jones95}.  Intuitively, a $S$
improves $P$ if every satisfiable ground instance of $P$ is also a ground instance of $S\,P$.  Jones
uses improving substitutions to refine the results of type inference while still inferring principal
types.  We will adopt a similar approach, but in typing instead of type inference.

\paragraph{Syntax.}

We begin by extending the syntax of class axioms to include functional dependency assertions:
\[\begin{array}{lrr@{\hspace{5px}}c@{\hspace{5px}}l}
  \multicolumn{2}{l}{\text{Index sets}} & X,Y & \subseteq & \mathbb{N} \\
  \multicolumn{2}{l}{\text{Class axioms}} & \alpha & ::= & \fd{C}{X}{Y} \mid \qclause{d}{\seq{t}}{P}{\pi} \\
\end{array}\]
In the representation of functional dependency axioms, we treat the class parameters by index rather
than by name.  If $A$ were the axioms for the example above, we would expect to have a dependency
\[
  \fd{\mathtt{Elems}}{\{0\}}{\{1\}} \in A.
\]
Any particular class name may appear in many functional dependency assertions, or in none at all.
We adopt some notational abbreviations: if $X$ is an index set, we write $\pi =_X \pi'$ to indicate
that $\pi$ and $\pi'$ agree at least on those parameters with indices in $X$, and similarly write
$\pi \stackrel{S}{\sim}_X \pi'$ to indicate that $S$ is a unifier for those parameters of $\pi$ and
$\pi'$ with indices in $X$.

\paragraph{Improvement.}

To account for improvement in typing, we need a syntactic characterization of improving
substitutions. In the case of functional dependencies, this can be given quite directly.  We can
give an improvement rule as a direct translation of the intuitive description above:
\[
\infbox{\irule[\trule{Fundep}]
              {\begin{array}{c}
                 P \entails \predh{C}{\seq\tau} \isp
                 P \entails \predh{C}{\seq\upsilon} \\[-2px]
                 (\fd{C}{X}{Y}) \in A \isp
                 \seq\tau =_X \seq\upsilon \isp
                 \seq\tau \stackrel{S}{\sim}_Y \seq\upsilon
               \end{array}};
              {\impr{A}{S}{P}}}
\]
For example, if we have some $Q$ such that $Q \entails \Elems \tau \upsilon$ and $Q \entails \Elems
\tau {\upsilon'}$, then \trule{Fundep} says that the any unifying substitution $U$ such that
$U\,\upsilon = U\,\upsilon'$ is an improving substitution for $Q$.  If $S$ is an improving
substitution for $P$, then the qualified type schemes $(P \mid \sigma)$ and $(S\,P \mid S\,\sigma)$
are equivalent, and we should be able to replace one with the other at will in typing derivations.
One direction is already possible: if a term has type $\sigma$, then it is always possible to use it
with type $S\,\sigma$ (by a suitable series of applications of \trule{\Iforall} and
\trule{\Eforall}).  On the other hand, there is not (in general) a way with our existing typing
rules to use a term of type $S\,\sigma$ as a term of type $\sigma$.  We add a typing rule to support
this case.
\[
\infbox{\irule[\trule{Impr}]
              {\typ{S\,P}{S\,\Gamma}{A}{M}{S\,\sigma}}
              {\impr{A}{S}{P}};
              {\tys{M}{\sigma}}}
\]
As in the case of \trule{\I\then} and \trule{\E\then}, \trule{Impr} has no effect on the semantics
of terms.  Thus, if we have a derivation
\[
  \infbox{\lproof[\De]
                 {\lproof[\De_1]
                         {\vdots};
                         {\typ{S\,P}{S\,\Gamma}{A}{M}{S\,\sigma}}}
                 {\assume{\impr{A}{S}{P}}};
                 {\tys{M}{\sigma}}}
\]
we define that
\(
  \interp{\De}S'\eta = \interp{\De_1}S''\eta,
\)
where $S'' \circ S = S'$ (the existence of such an $S''$ is guaranteed by the soundness of
\trule{Fundep}).  Finally, we add a rule to the equality judgment allowing us to use improving
substitutions in equality proofs.
\[
\infbox{\irule[\erule{Impr}]
              {\eql{S\,P}{S\,\Gamma}{\tuple{A,Si,Im}}{M}{N}{S\,\sigma}}
              {\impr{A}{S}{P}};
              {\eql{P}{\Gamma}{\tuple{A,Si,Im}}{M}{N}{\sigma}}}
\]

\paragraph{Validating Functional Dependency Axioms.}

We must augment the context rule to check that the axioms respect the declared dependencies.  This
can be accomplished by, first, refining the overlap check to assure that no axioms overlap on the
determining parameters of a functional dependencies, and second, requiring that, for each dependency
$\fd{C}{X}{Y}$ and each instance $P \then \pi$ of class $C$, any variables in the positions $Y$ are
determined by the functional dependencies of $P$.  Our formalization of the latter notion follows
Jones's development~\cite{Jones08}.  We define the closure of a set of variables $J$ with respect to
the functional dependencies $F$ as the least set $J^+_F$ such that
\begin{itemize}
\item $J \subseteq J^+_F$; and
\item If $U \tr V \in F$ and $U \subseteq J^+_F$, then $V \subseteq J^+_F$.
\end{itemize}
We write $ftv_X(\predh{C}{\seq\tau})$ to abbreviate $\bigcup_{x \in X}ftv(\tau_x)$, define the
instantiation of a functional dependency assertion $\fd{C}{X}{Y}$ at a predicate
$\pi = \predh{C}{\seq\tau}$, as the dependency $ftv_X(\pi) \tr ftv_Y(\pi),$ and write $fd(A,P)$ for the
set of the instantiation of each functional dependency assertion in $A$ at each predicate in $P$.
We can now define the verification conditions for axioms and the new version of \trule{Ctxt}, as
follows.
\begin{gather*}
\infbox{\irule{\{ \pi \nsim_X \pi' \mid (d: P \then \pi), (d': P' \then \pi'), (\fd{class(\pi)}{X}{Y}) \in A \}};
              {\nono{A}}}
\\
\infbox{\irule{\{ ftv(\pi_Y) \subseteq ftv(\pi_X)^+_{fd(A,P)} \mid (d:P \then \pi),(\fd{class(\pi)}{X}{Y}) \in A \}};
              {\cov{A}}}
\\
\infbox{\irule[\trule{Ctxt}]
              {\!\!\!\!
               \begin{array}{c}
                 \nono{A}
                 \isp
                 \cov{A}
                 \\[1px]
                 \{ (\typ{P}{\Gamma, \overline{x_i : \sigma_{x_i}}}{A}{Im(y,d)}{\sigma_{y,d}}) \mid \tuple{y,d} \in \dom(Im) \}
                 \\[1px]
                 \typ{P}{\Gamma, \overline{x_i : \sigma_{x_i}}}{A}{M}{\sigma}
               \end{array}};
              {\typ{P}{\Gamma}{\tuple{A,Si,Im}}{M}{\sigma}}}
\end{gather*}

\subsection{Soundness}

The significant challenge in proving soundness of the extended rules is showing that when
$\impr{A}{S}{P}$ is derivable, $S$ is an improving substitution for $P$.  Once we have established
that result, the remaining soundness results will be direct.  We introduce notation for the
satisfiable ground instances of predicates~$P$:
\[
  \gr{P}_A = \{ S\,P \mid S \in GSubst(ftv(P)), \entails_A S\,P \}.
\]
We can now formally describe an improving substitution.
\begin{lemma}\label{thm:sound-impr}
  Given a set of axioms~$A$ such that \nono{A} and \cov{A}, if $\impr{A}{S}{P}$, then $\gr{P}_A =
  \gr{S\,P}_A$.
\end{lemma}

\begin{proof}
By contradiction.  Assume that $\impr{A}{S}{P}$; then we must have $\pi_0,\pi_1$ such that
$\entails_A \pi_0, \entails_A \pi_1$ and there is a functional dependency
\(
  (\fd{class(\pi)}{X}{Y}) \in A
\)
such that $\pi_0 =_X \pi_1$ but $\pi_0 \not=_Y \pi_1$.  We proceed by induction on the heights of
the derivations of $\entails_A \pi_0, \entails_A \pi_1$.
\begin{itemize}
\item There are distinct axioms $\clause{d}{P}{\pi_0'},\clause{d'}{P'}{\pi_1'} \in A$ and
  substitutions $S_0,S_1$ such that $S_0\,\pi_0' = \pi_0$ and $S_1\,\pi_1' = \pi_1'$.  But then $S_0
  \circ S_1$ is a unifier for $\pi_0' \sim_X \pi_1'$, contradicting $\nono{A}$.
\item There is a single axiom $\clause{d}{P}{\pi_0'}$ and substitutions $S_0,S_1$ such that
  $S_0\,\pi_0' = \pi_0$ and $S_1\,\pi_0' = \pi_1$.  We identify two sub-cases.
  \begin{itemize}
  \item There is some type variable in $ftv_Y(\pi_0') \setminus ftv_X(\pi_0')$ that is not
    constrained by $P$.  This contradicts $\cov{A}$.
  \item There is some $\pi' \in P$ such that $S_0\,\pi'$ and $S_1\,\pi'$ violate a functional
    dependency of $class(\pi')$.  The derivations of $\entails S_0\,\pi'$ and $\entails S_1\,\pi'$
    must be shorter than the derivations of $\entails \pi_0,\entails \pi_1$, and so we have the desired result by
    induction. \qedhere
  \end{itemize}
\end{itemize}
\end{proof}

\begin{theorem}[Soundness of typing]\label{thm:type-sound-impr}
   Given a class context $\Psi$, if $\De$ is a derivation of $P \mid \Gamma \vdash_\Psi M : \sigma$,
   $S$ is a substitution, and $\eta$ is an $(S\,\Gamma)$-environment, then $\interp{\De}S\eta \in
   \interp{(S\,P \mid S\,\sigma)}_\Psi$.
\end{theorem}

\begin{proof}
  We need only consider the \trule{Impr} case.  From Lemma~\ref{thm:sound-impr}, we have that if $T$
  improves $P$, then $\interp{(P \mid \sigma)}_\Psi = \interp{(T\,P \mid T\,\sigma)}_\Psi$, and so
  the result follows from the induction hypothesis.
\end{proof}

We extend our notion of ambiguity to take account of functional dependencies: it is enough for the
variables in the predicates~$P$ to be determined by the variables of $\tau$.

\begin{defn}\label{def:ambiguous-impr}
  A type scheme $\sigma = \forall \vec t. P \then \tau$ is unambiguous (given class axioms~$A$) if
  $ftv(P) \subseteq ftv(\tau)^+_{fd(A,P)}$.
\end{defn}

\noindent
The previous definition of ambiguity is a special case of this definition, where $fd(A,P)$ is always
empty.  As uses of \trule{Impr} do not affect the semantics of terms, its introduction does not
compromise coherence.

\begin{theorem}\label{thm:coherence-impr}
  If $\sigma$ is unambiguous and $\De_1,\De_2$ are derivations of $P \mid \Gamma \vdash_\Psi M:
  \sigma$, then $\interp{\De_1} = \interp{\De_2}$.
\end{theorem}

\begin{theorem}[Soundness of $\equiv$]\label{thm:equiv-sound-impr}
  Given a class context $\Psi$, if $\sigma$ is unambiguous, $\eqs{M}{N}{\sigma}$, and $\De_M,\De_N$
  are derivations of $\typ{P}{\Gamma}{\Psi}{M}{\sigma},\typ{P}{\Gamma}{\Psi}{N}{\sigma}$, then
  $\interp{\De_M}_\Psi = \interp{\De_N}_\Psi$.
\end{theorem}

\begin{proof}
  Again, we need consider only the \erule{Impr} case.  Without loss of generality, assume $\De_M$
  and $\De_N$ are by \trule{Impr}, with subderivations $\De_M'$ and $\De_N'$.  As the
  interpretations of $\De_M$ and $\De_N$ are equal to the interpretations of $\De_M'$ and $\De_N'$,
  the result follows from the induction hypothesis.
\end{proof}

\section{Related Work}\label{sec:related}

The semantics of polymorphism, in its various forms, has been studied extensively over the past half
century; however, the particular extensions of Haskell that motivated this work are recent, and have
received little formal attention.

Our approach was inspired by Ohori's semantics of Core ML~\cite{Ohori89}.  While Ohori's approach
describes the semantics of polymorphism, he does not represent polymorphic values directly, which
leads to an unusual treatment of the typing of !let! expressions.  Harrison extends Ohori's approach
to treat polymorphic recursion~\cite{Harrison05}; in doing so, he provides a representation of
polymorphic values.  Harrison suggests that his approach could be applied to type classes as well.

Ohori's approach to the semantics of ML is somewhat unusual; more typical approaches include those
of Milner~\cite{Milner78} and Mitchell and Harper~\cite{MitchellHarper88}.
Ohori identifies reasons to prefer his approach over either
that of Milner or that of Mitchell and Harper: both approaches use a semantic domain with far more
values than correspond to values of ML, either because (in the untyped case) those values would not
be well-typed, or (in the explicit typed case) they differ only in the type-level operations.

The semantics of type-class-based overloading has also received significant attention.  Wadler and
Blott~\cite{WadlerBlott89} described the meaning of type classes using a dictionary-passing
translation, in which overloaded expressions are parameterized by type-specific implementations of
class methods.  Applying their approach to the full Haskell language, however, requires a target
language with more complex types than their source language.  For example, in translating the
!Monad! class from the Haskell prelude,
the dictionary for \predh{\texttt{Monad}}{\tau} must contain polymorphic values for the !return! and
!(>>=)! methods.

In his system of qualified types~\cite{Jones92}, Jones generalized the treatment of evidence by
translating from a language with overloading (OML) to a language with explicit evidence abstraction
and application.
Jones does not provide a semantics of the language with
explicit evidence abstraction and application; indeed, such a semantics could not usefully be
defined without choosing a particular form of predicate, and thus a particular form of evidence.

Odersky, Wadler and Wehr~\cite{Odersky95} propose an alternative formulation of overloading,
including a type system and type inference algorithm, and a ideal-based semantics of qualified
types.  However, their approach requires a substantial restriction to the types of overloaded
values
which rules out many functions in the Haskell prelude
as well as the examples from our previous work~\cite{Morris10}.

\newcommand{\SystemFC}{System~$\mathrm{F_C}$}

Jones~\cite{Jones00} introduced functional dependencies in type classes, and discusses their use to
improve type inference; his presentation of improvement is similar to ours, but he does not augment
typing as does our \trule{Impr} rule.  Sulzmann et al.~\cite{Sulzmann07} give an alternative
approach to the interaction of functional dependencies and type inference, via a translation into
constraint-handling rules; unfortunately, their presentation conflates properties of their
translation, such as termination, with properties of the relations themselves.
\SystemFC{}~\cite{Sulzmann07b} extends System~F with type-level equality constraints and
corresponding coercion terms.  While we are not aware of any formal presentation of functional
dependencies in terms of \SystemFC{}, we believe that a formulation of our \trule{Fundep} rule in
terms of equality constraints is possible.  In contrast to our approach, \SystemFC{} requires
extending the domain of the semantics, while still requiring translation of source-level features
(functional dependencies or GADTs) into features of the semantics (equality constraints).

\section{Conclusion}\label{sec:conclusion}

We have proposed an alternative approach to the semantics of overloading, based on interpreting
polymorphic values as sets of their monomorphic interpretations, which avoids several problems with
traditional translation-based approaches.  We have applied this result to a simple overloaded
calculus, and shown the soundness of its typing and equality judgments.  Finally, we have argued
that the approach is flexible enough to support extensions to the type system, such as allowing the
use of improving substitutions in typing.  We conclude by identifying directions for future
work:
\begin{itemize}
\item Practical class systems are richer than the one used in this paper.  We would like to extend
  these results to fuller systems, including our prior work on instance chains.
\item Dictionary-passing provides both a semantics of overloading and an implementation technique.
  We would like to explore whether implementation techniques based on specialization can be used to
  compile practical languages.
\item We claim that our approach avoids making distinctions between some observationally equivalent
  terms (such as in the polymorphic identity function example).  We would like to explore whether
  adequacy and full abstraction results for the underlying frame model can be extended to similar
  results for our semantics.
\item Our definition of equality provides $\eta$-equivalence; however, $\eta$ equivalence is not
  sound for Haskell.  We would like to explore either whether our approach can be adapted to a
  language without $\eta$-equivalence.
\end{itemize}

\paragraph{Acknowledgments.}

We would like to thank: Mark Jones for initially suggesting Ohori's semantics of ML polymorphism as
a basis for understanding overloading; Jim Hook for proposing the polymorphic identity function
example; and, Keiko Nakata for her helpful feedback on drafts of the paper.

\bibliographystyle{abbrvnat}
\bibliography{main}

\appendix

\end{document}